\RequirePackage{amsmath}
\documentclass{llncs}

\usepackage{multicol}

\usepackage{amsmath}
\usepackage{amsfonts}
\usepackage{amssymb}
\usepackage{verbatim}

\usepackage{gastex}
\usepackage{graphicx}
\usepackage{epic}
\usepackage{eepic}
\usepackage{epsfig,float}
\usepackage{pdfsync}

\usepackage{multicol}
\newcommand{\tid}{\mbox{{\bf 1}}}
\renewcommand{\le}{\leqslant}
\renewcommand{\ge}{\geqslant}

\newcommand{\ol}{\overline}
\newcommand{\eps}{\varepsilon}
\newcommand{\emp}{\emptyset}

\newcommand{\Sig}{\Sigma}

\newcommand{\noin}{\noindent}

\newcommand{\bi}{\begin{itemize}}
\newcommand{\ei}{\end{itemize}}
\newcommand{\be}{\begin{enumerate}}
\newcommand{\ee}{\end{enumerate}}
\newcommand{\bd}{\begin{description}}
\newcommand{\ed}{\end{description}}
\newcommand{\bq}{\begin{quote}}
\newcommand{\eq}{\end{quote}}

\newcommand{\cD}{{\mathcal D}}

\newcommand{\one}{{\mathbf 1}}

\spnewtheorem{conj}{Conjecture}{\bfseries}{\rmfamily}

\newcommand{\distlemma}{the Distinguishability Lemma}

\title{Quotient Complexities of Atoms in Regular Ideal Languages \thanks{This work was supported by the Natural Sciences and Engineering Research Council of Canada under grant No.~OGP0000871.}
}

\author{Janusz~Brzozowski\inst{1} \and Sylvie Davies\inst{2}}

\titlerunning{Quotient Complexities of Atoms in Regular Ideal Languages}

\authorrunning{J. Brzozowski and S. Davies}   

\institute{David R. Cheriton School of Computer Science, University of Waterloo \\
Waterloo, ON, Canada N2L 3G1\\
{\tt brzozo@uwaterloo.ca}
\and
Department of Pure Mathematics, University of Waterloo \\
Waterloo, ON, Canada N2L 3G1\\
{\tt sldavies@uwaterloo.ca}
}
\begin{document}

\maketitle
\begin{abstract}
A (left) quotient of a language $L$ by a word $w$ is the language $w^{-1}L=\{x\mid wx\in L\}$.
The quotient complexity of a regular language $L$ is the number of quotients of $L$; it is equal to the state complexity of $L$, which is the number of states in a minimal deterministic finite automaton accepting $L$.
An atom of $L$  is an equivalence class of the relation in which two words are equivalent if for each quotient, they either are both in the quotient or both not in it; 
hence  it is a non-empty intersection of complemented and uncomplemented  quotients of $L$.
A right (respectively, left and two-sided) ideal is a language $L$ over an alphabet $\Sig$ that satisfies $L=L\Sig^*$ (respectively, $L=\Sig^*L$ and $L=\Sig^*L\Sig^*$).
We compute the maximal number of atoms and  the maximal quotient complexities of atoms of right, left and two-sided regular ideals.
\medskip

\noin

\end{abstract}

\section{Introduction}

We assume that the reader is familiar with basic concepts of regular languages and finite automata; more background is given in the next section.
Consider a regular language $L$ over a finite non-empty alphabet $\Sig$.
Let $\cD=(Q,\Sig,\delta,q_1, F)$ be a minimal \emph{deterministic finite automaton (DFA)} recognizing $L$, where $Q$ is the set of \emph{states}, $\delta\colon Q\times \Sig\to Q$ is the \emph{transition function}, $q_1$ is the \emph{initial} state, and $F\subseteq Q$ is the set of \emph{final} states.
There are three natural equivalence relations associated with $L$ and $\cD$. 

The  \emph{Nerode right congruence}~\cite{Ner58} is defined as follows: Two words $x$ and $y$ are equivalent if for every $v\in \Sig^*$, $xv$ is in $L$ if and only if $yv$ is in $L$. 
The set of all words that ``can follow'' a given word $x$ in $L$ is the \emph{left quotient of $L$ by $x$}, defined by $x^{-1}L=\{v\mid vx\in L\}$. 
In automaton-theoretic terms $x^{-1}L$ is the set of all words $v$ that are accepted from the state $q=\delta(q_1,x)$ reached when $x$ is applied to the initial state of $\cD$; this is  known as the \emph{right language} of state~$q$, the language accepted by  DFA $\cD_q=(Q,\Sig,\delta,q, F)$. The Nerode equivalence class containing $x$ is  known as the \emph{left language} of state~$q$, the language accepted by DFA ${}_q\cD=(Q,\Sig,\delta, q_1, \{q\})$.
The number $n$ of Nerode equivalence classes is the number of distinct left quotients of $L$,  known as its \emph{quotient complexity}~\cite{Brz10}. This is the same number as the number of states in $\cD$, and is therefore known as $L$'s \emph{state complexity}~\cite{Yu01}. 
Quotient/state complexity is now a commonly used measure of complexity of a regular language, and constitutes a basic reference for other measures of complexity. One can also define the quotient complexity of a Nerode equivalence class, that is, of the language accepted by  DFA ${}_q\cD$. In the worst case -- for example, if $\cD$ is strongly connected --  this is $n$ for every $q$.

The \emph{Myhill congruence}~\cite{Myh57} refines the Nerode right congruence and is a (two-sided) congruence. Here word $x$ is equivalent to word $y$ if for all $u$ and $v$ in $\Sig^*$, $uxv$ is in $L$ if and only if $uyv$ is in $L$. This is also known as the \emph{syntactic congruence}~\cite{Pin97} of $L$. The quotient set of $\Sig^+$ by this congruence is the \emph{syntactic semigroup} of $L$. 
In automaton-theoretic terms two words are equivalent if they induce the same transformation of the set of states of a minimal DFA of $L$.
The quotient complexity of Myhill classes has not been studied.

The third equivalence, which we call the \emph{atom congruence} is a left congruence refined by the Myhill congruence. Here two words $x$ and $y$ are equivalent if 
 $ux\in L$ if and only if  $uy\in L$ for all $u\in \Sig^*$. 
 Thus $x$ and $y$ are equivalent if
 $x\in u^{-1}L$ if and only if $y\in u^{-1}L$.
 An equivalence class of this relation is called an \emph{atom} of $L$~\cite{BrTa14}. 
It follows that an atom is a non-empty intersection of complemented and uncomplemented quotients of $L$.

This congruence is related to the Myhill and Nerode congruences in a natural way. Say a congruence on $\Sig^*$ \emph{recognizes} $L$ if $L$ can be written as a union of the congruence's classes. The Myhill congruence is the unique \emph{coarsest} congruence (that is, the one with the fewest equivalence classes) that recognizes $L$~\cite{Pin97}. The Nerode and atom congruences are respectively the coarsest \emph{right} and \emph{left} congruences that recognize $L$.

The quotient complexity of atoms of regular languages has been studied in~\cite{BrDa14,BrTa13,Iva14}.
In this paper we study the quotient complexity of atoms in three subclasses of regular languages, namely, right, left, and two-sided ideals.

Ideals are fundamental concepts in semigroup theory. A language $L$ over an alphabet $\Sig$ is a 
\emph{right} (respectively, \emph{left} and \emph{two-sided}) \emph{ideal} if $L=L\Sig^*$ (respectively, $L=\Sig^*L$ and $L=\Sig^*L\Sig^*$).
The quotient complexity of regular ideal languages has been studied in~\cite{BJL13}, and the reader should refer to that paper for more information about ideals.
Ideals appear in  pattern matching.  A right (left) ideal $L\Sig^*$ ($\Sig^*L$) represents the set of all words beginning (ending) with some word of a  given set $L$, and $\Sig^*L\Sig^*$ is the set of all words containing a factor from $L$.

\section{Preliminaries}

It is well known that a language $L\subseteq \Sig^*$ is regular if and only if it has a finite number of quotients. We denote the number of quotients of $L$  (the \emph{quotient complexity}) by $\kappa(L)$. This is the same as the \emph{state complexity}, the number of states in a minimal DFA of $L$. Since we will not be discussing other measures of complexity, we refer to both quotient and state complexity as just \emph{complexity}.

Let the set of quotients of a regular language $L$ be $K=\{K_1,\dots,K_n\}$.
The \emph{quotient automaton} of $L$ is the DFA $\cD=(K,\Sig,\delta,L, F)$, where 
$\delta(K_i,a)=K_j$ if $a^{-1}K_i=K_j$, $L=K_1=\eps^{-1}L$ by convention, and $F=\{K_i\mid \eps \in K_i\}$. This DFA is uniquely defined by $L$ and is isomorphic to every minimal DFA of $L$.

A \emph{transformation} of a set $Q_n$ of $n$ elements is a mapping of $Q_n$ \emph{into} itself, whereas a \emph{permutation}
of $Q_n$ is a mapping of $Q_n$ \emph{onto} itself.
In this paper we consider only transformations of finite sets, and we assume
without loss of generality  that $Q_n=\{1,\ldots, n\}$.
An arbitrary transformation has the form
$$
t=\left( \begin{array}{ccccc}
1 & 2 &   \cdots &  n-1 & n \\
i_1 & i_2 &   \cdots &  i_{n-1} & i_{n}
\end{array} \right ),
$$
where $i_k\in Q_n$ for $1\le k\le n$.
The image of element $i$ under transformation $t$ is denoted by $it$.
The image of $S \subseteq Q_n$ is $St = \cup_{i\in S}\{it\}$.
The \emph{identity} transformation $\tid$ maps each element to itself.
For $k\ge 2$, a transformation (permutation) $t$ is a \emph{$k$-cycle} if there is a set $P=\{q_1,q_2,\ldots,q_{k}\} \subseteq Q_n$ such that if $q_1t=q_2, q_2t=q_3,\ldots,q_{k-1}t=q_{k},q_{k}t=q_1$, and $qt = q$ for all $q \not\in P$.
A $k$-cycle is denoted by $(q_1,q_2,\ldots,q_{k})$.
A~2-cycle $(q_1,q_2)$ is called a \emph{transposition}.
A transformation is \emph{constant} if it maps all states to a single state $q$; we denote it by $(Q_n\to q)$.
A  transformation $t$ is \emph{unitary}  if $p\neq q$, $pt =q$ and $rt=r$ for all $r\neq p$; we denote it by $(p\to q)$. 
The following is well-known:

\begin{proposition}
\label{prop:piccard}
The complete transformation monoid $T_n$ of size $n^n$ can be generated by any generators of the symmetric group $S_n$ (the group of all permutations of $Q_n$) together with a unitary transformation. In particular, $T_n$ can be generated by $\{(1,\dotsc, n),(1,2),(n \to 1)\}$, and by $\{(1,\dotsc,n),(2,\dotsc,n),(n \to 1)\}$.
\end{proposition}

For a DFA $\cD = (Q,\Sig,\delta,q_1,F)$ we define the transformations $\{\delta_w \mid w \in \Sig^+\}$ by $q\delta_a = \delta(q,a)$ for $a \in \Sig^*$, and $q\delta_{w} = q\delta_x\delta_a$ for $w = xa$. This set is a semigroup under composition and it is called the \emph{transition semigroup} of $\cD$. 
The transformation $\delta_w$ is called the \emph{transformation induced by $w$}. To simplify notation, we usually make no distinction between the word $w \in \Sig^+$ and the transformation $\delta_w$. 
If $\cD$ is the quotient automaton of $L$, then the transition semigroup of $\cD$ is isomorphic to the syntactic semigroup of $L$~\cite{Pin97}.
A state $q \in Q$ is \emph{reachable from $p \in Q$} if $pw = q$ for some $w \in \Sig^+$, and \emph{reachable} if it is reachable from $q_1$.
Two states $p,q$ are \emph{indistinguishable} if $pw \in F \Leftrightarrow qw \in F$ for all $w \in \Sig^+$, and \emph{distinguishable} otherwise.
Indistinguishability is an equivalence relation on $Q$; furthermore, if $\cD$ recognizes a language $L$, we can compute $\kappa(L)$ by counting the number of equivalence classes under  indistinguishability of the reachable states of $\cD$.
A state is \emph{empty} if its right language (defined in the introduction) is $\emp$.

\section{Atoms}
Atoms of regular languages were studied in~\cite{BrTa14}, and their complexities  in~\cite{BrDa14a,BrTa13}.
As discussed earlier, atoms are the classes of the \emph{atom congruence}, a left congruence which is the natural counterpart of the Myhill two-sided congruence and Nerode right congruence. 
The Myhill and Nerode congruences are fundamental in regular language theory, but it seems comparatively little attention has been paid to the atom congruence and its classes.
In~\cite{Brz13} it was argued that it is useful to consider the complexity of a language's atoms when searching for highly complex regular languages, since one would expect such languages to have highly complex atoms.

Below we present an alternative characterization of atoms, which we use in our proofs. Earlier papers on atoms such as~\cite{BrDa14a,BrTa13,BrTa14} take this as the definition of atoms, for it was not known until recently that atoms may be viewed as congruence classes (this fact was first noticed by Iv\'an in~\cite{Iva14}).

From now on assume all languages are non-empty.
Denote the complement of a language $L$ by $\ol{L} = \Sig^* \setminus L$.
Let $Q_n=\{1,\dots,n\}$ and let $L$ be a regular language with quotients $K = \{K_1,\dotsc,K_n\}$. Each subset $S$ of $Q_n$ defines an \emph{atomic intersection} $A_S = \bigcap_{i \in S} K_i \cap \bigcap_{i \in \ol{S}} \ol{K_i}$, where $\ol{S} = Q_n \setminus S$.
An \emph{atom} of $L$ is a non-empty atomic intersection. 
Since atoms are pairwise disjoint,  every atom $A$ has a unique atomic intersection associated with it, and this atomic intersection has a unique subset $S$ of $K$ associated with it. This set $S$ is called the \emph{basis} of $A$.

Throughout the paper, $L$ is a regular language of complexity $n$ with quotients $K_1,\dotsc,K_n$ and minimal DFA $\cD = (Q_n,\Sig,\delta,1,F)$ such that the language of state $i$ is $K_i$. 
Let $A_S = \bigcap_{i \in S} K_i \cap \bigcap_{i \in \ol{S}} \ol{K_i}$ be an atom. For any $w\in\Sig^*$ we have 
$$w^{-1}A_S = \bigcap_{i \in S}w^{-1} K_i \cap \bigcap_{i \in \ol{S}} \ol{w^{-1} K_i}.$$
Since a quotient of a quotient of $L$ is also a quotient of $L$, $w^{-1}A_S$ has the form;
$$w^{-1} A_S = \bigcap_{i \in X} K_i \cap \bigcap_{i \in Y} \ol{K_i},$$
where $|X| \le |S|$ and $|Y| \le n-|S|$, $X,Y \subseteq Q_n$.

The complexity of atoms of a regular language was computed in~\cite{BrTa13} using a unique NFA defined by $L_n$, called the \emph{\'atomaton}. In that NFA the language of each state $q_S$ is an atom $A_S$ of $L_n$. To find the complexity of that atom, the \'atomaton started in state $q_S$ was converted to an equivalent DFA. 
A more direct and simpler method was used by Szabolcs Iv\'an~\cite{Iva14} who constructed the DFA for the atom directly from the DFA $\cD_n$. We follow that approach here and outline it briefly for completeness. 

For any regular language $L$ an atom $A_S$ corresponds to the ordered pair $(S,\ol{S}$), where $S$ ($\ol{S}$) is the set of subscripts of uncomplemented (complemented) quotients. If $L$ is represented by a DFA $\cD=(Q,\Sig,\delta, q_1, F)$, it is more convenient to think of $S$ and $\ol{S}$ as subsets of $Q$.
Similarly, any quotient of $A_S$ corresponds to a pair $(X,Y)$ of subsets of $Q$. 
For the quotient of $A_S$ reached when a letter $a\in \Sig$ is applied to the quotient corresponding to $(X,Y)$ we get
\[ a^{-1} \left(\bigcap_{i \in X} K_i \cap \bigcap_{i \in Y} \ol{K_i}\right)
= \bigcap_{i \in X} a^{-1} K_i \cap \bigcap_{i \in Y} \ol{a^{-1}K_i}
= \bigcap_{i \in X} K_{ia} \cap \bigcap_{i \in Y} \ol{K_{ia}}. \]
In terms of pairs of subsets of $Q$, from $(X,Y)$ we reach $(Xa,Ya)$. 
Note that if $X\cap Y\neq \emp$ in $(X,Y)$ then the corresponding quotient is empty.
Note also that the quotient of atom $A_S$ corresponding to $(X,Y)$ is final if and only if
each quotient $K_i$ with $i\in X$ contains $\eps$, and each $K_j$ with $j \in Y$ does not contain $\eps$.

These considerations lead to the following definition of a DFA for $A_S$.
\begin{definition}
Suppose $\cD=(Q,\Sig,\delta, q_1, F)$ is a DFA and let $S \subseteq Q$.
Define the DFA $\cD_S = (Q_S,\Sig,\Delta,(S,\ol{S}),F_S)$, where
\bi
\item
$Q_S = \{(X,Y) \mid X,Y \subseteq Q, X \cap Y = \emp\} \cup \{\bot\}$.
\item
For all $a \in \Sig$, $\Delta((X,Y),a) = (\delta(X,a),\delta(Y,a))$ if $\delta(X,a) \cap \delta(Y,a) \ne \emp$, and $\Delta((X,Y),a) = \bot$ otherwise; and $\Delta(\bot,a) = \bot$.
\item
$F_S = \{(X,Y) \mid X\subseteq F, Y \subseteq \ol{F}\}$. 
\ei
\end{definition}
DFA $\cD_S$ recognizes the atomic intersection $A_S$ of $L$. If $\cD_S$ recognizes a non-empty language, 
then $A_S$ is an atom. 

\section{Complexity of Atoms in Regular Languages}

Upper bounds on the maximal complexity of atoms of regular languages were derived in~\cite{BrTa13}; for completeness we include these results.
For $n=1$ there is only one non-empty language $L=\Sig^*$; it has one atom, $L$, which is of complexity~1. From now on assume that $n\ge 2$.

\begin{proposition}
\label{prop:reg}
Let $L$ be a regular language with $n\ge 2$ quotients. Then $L$ has at most $2^n$ atoms.
If $S\in \{Q_n,\emp\}$, then $\kappa(A_S) \le 2^n-1$ quotients. Otherwise,
\[ \kappa(A_S) \le 1+ \sum_{x=1}^{|S|}\sum_{y=1}^{n-|S|} \binom{n}{x}\binom{n-x}{y}. \]
\end{proposition}
\begin{proof}
Since the number of subsets $S$ of $Q_n$  is $2^n$, there are at most that many atoms.
For atom complexity consider the following three cases:
\be
\item
$S=Q_n$.
Then  $A_{Q_n}= \bigcap_{i \in Q_n} K_i$ is the intersection of all quotients of $L$. 
For $w\in \Sig^*$,  $w^{-1}A_{Q_n}=\bigcap_{i \in X} K_i$, where $1\le |X|\le |Q_n|$. 
Hence there are at most $2^n-1$ quotients of this atom.
\item
$S=\emp$. 
 Now $A_\emp=\bigcap_{i \in Q_n} \ol{K_i}$, and
$w^{-1}A_\emp=\bigcap_{i \in Y} \ol{K_i}$, where 
$1\le |Y|\le |Q_n|$.
As in the first case, there are at most  $2^n-1$ quotients of this atom.
\item
$\emp \subsetneq S \subsetneq Q_n$. Then $A_S = \bigcap_{i \in S} K_i \cap \bigcap_{i \in \ol{S}} \ol{K_i}$.
Every quotient of $A_S$ has the form $w^{-1} A_S = \bigcap_{i \in X} K_i \cap \bigcap_{i \in Y} \ol{K_i},$ where $1\le |X| \le |S|$ and $1\le |Y| \le n-|S|$.
There are two subcases:
	\be
	\item
	If $X\cap Y\neq \emp$, then $w^{-1}A_S=\emp$.
	\item
	If $X\cap Y = \emp$, there are at most 
	$\sum_{x=1}^{|S|}\sum_{y=1}^{n-|S|} \binom{n}{x}\binom{n-x}{y} $
	quotients of $A_S$ of this form. This follows since $\binom{n}{x}$ is the number of ways to choose a set $X \subseteq Q_n$ of size $x$, and once $X$ is fixed, $\binom{n-x}{y}$ is the number of ways to choose a set $Y \subseteq Q_n$ of size $y$ that is \emph{disjoint} from $X$. Taking the sum over the permissible values of $x$ and $y$ gives the formula above.
	\ee
Adding the results of (a) and (b)  we have the required bound. \qed
\ee
\end{proof}

 It was shown in~\cite{Brz13} that the language $L_n$ accepted by the minimal DFA $\cD_n$ of Definition~\ref{def:reg}, also illustrated  in Figure~\ref{fig:reg}, meets all the complexity bounds for common operations on regular languages.   
 \begin{definition}
\label{def:reg}
For $n\ge 2$, let $\cD_n=(Q_n,\Sig,\delta_n, 1, \{n\})$, where 
$Q_n=\{1,\dots,n\}$
is the set of states,
$\Sig=\{a,b,c\}$ is the alphabet, the transition function $\delta_n$ is defined 
by
$a = (1,\dots,n)$,
$b = (1,2)$,
and $c = (n\rightarrow 1)$,   state 1 is the initial state, and $\{n\}$ is the set of final states.
Let $L_n$ be the language accepted by~$\cD_n$.
(If $n=2$, $a$ and $b$ induce the same transformation; hence $\Sig=\{a,c\}$ suffices.)
\end{definition}

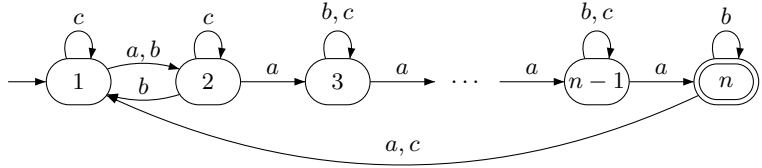
\begin{figure}[th]
\unitlength 7pt
\begin{center}\begin{picture}(37,10)(0,1)
\gasset{Nh=2.5,Nw=3.5,Nmr=1.25,ELdist=0.4,loopdiam=1.5}
\node(1)(1,7){1}\imark(1)
\node(2)(8,7){2}
\node(3)(15,7){3}
\node[Nframe=n](3dots)(22,7){$\dots$}
	{\small
\node(n-1)(29,7){$n-1$}
	}
	{\small
\node(n)(36,7){$n$}\rmark(n)
	}
\drawloop(1){$c$}
\drawedge[curvedepth= 1,ELdist=.1](1,2){$a,b$}
\drawedge[curvedepth= 1,ELdist=-1.2](2,1){$b$}
\drawloop(2){$c$}
\drawedge(2,3){$a$}
\drawloop(3){$b,c$}
\drawedge(3,3dots){$a$}
\drawedge(3dots,n-1){$a$}
\drawloop(n-1){$b,c$}
\drawedge(n-1,n){$a$}
\drawedge[curvedepth= 4.5,ELdist=-1.3](n,1){$a,c$}
\drawloop(n){$b$}
\end{picture}\end{center}
\caption{ DFA  of a regular language whose atoms meet the bounds.}
\label{fig:reg}
\end{figure}

It was proved in \cite{BrTa13} that $L_n$ has $2^n$ atoms, all of which are as complex as possible. We include the proof of this theorem following~\cite{Iva14}. We first prove a general result about distinguishability of states in $\cD_S$, which we will use throughout the paper.

\begin{lemma}[Distinguishability]
\label{lem:dist}
Let $\cD = (Q,\Sig,\delta,q_1,F)$ be a minimal DFA and for $S \subseteq Q$, let $\cD_S = (Q_S,\Sig,\Delta,(S,\ol{S}),F_S)$ be the DFA of the atom $A_S$. Then:
\be
\item
States $(X,Y)$ and $(X',Y')$ of $\cD_S$ are distinguishable 
if $X \ne X'$ and $A_{X},A_{X'}$ are both atoms, or if $Y \ne Y'$ and $A_{\ol{Y}},A_{\ol{Y'}}$ are both atoms.
\item
If one of $A_X$ or $A_{\ol{Y}}$ is an atom, then $(X,Y)$ is distinguishable from $\bot$.
\ee
\end{lemma}
\begin{proof}
First note that if $A_Z$ is an atom, then the initial state of $\cD_Z$ must be non-empty, so there is a word $w_Z$ such that $(Z,\ol{Z})w_Z = (U,V)$ with $U \subseteq F$, $V \subseteq \ol{F}$, i.e., $(U,V) \in F_S$. In particular, $(X,Y)w_X \in F_S$, since $Y \subseteq \ol{X}$. We also have $(X,Y)w_{\ol{Y}} \in F_S$, since $Y$ is sent to a subset of $\ol{F}$, and $X \subseteq \ol{Y}$ is sent to a subset of $F$. This proves (2): if one of $A_{X}$ or $A_{\ol{Y}}$ is an atom, then one of $w_X$ or $w_{\ol{Y}}$ is in the transition semigroup of $\cD$, and hence $(X,Y)$ can be mapped to a final state but $\bot$ cannot.
Now, we consider the two cases from (1):
\be
\item
$X \ne X'$. Suppose $X' \not\subseteq X$. Then $(X,Y)w_X \in F_S$, but $(X',Y')w_X \not\in F_S$, since $X' \setminus X$ is a non-empty subset of $\ol{X}$ and hence gets mapped outside of $F$. Thus $w_X$ distinguishes these states. If instead we have $X \not\subseteq X'$, then $w_{X'}$ distinguishes the states. Hence if $A_X,A_{X'}$ are atoms, $w_X$ and $w_{X'}$ are in the transition semigroup of $\cD$, and the states are distinguishable.
\item
$Y \ne Y'$. If $Y' \not\subseteq Y$, then $w_{\ol{Y}}$ distinguishes $(X,Y)$ from $(X',Y')$; otherwise, $w_{\ol{Y'}}$ distinguishes the states. As before, if $A_{\ol{Y}},A_{\ol{Y'}}$ are atoms then the states are distinguishable. \qed
\ee
\end{proof}

\begin{theorem}
For $n\ge 2$, the language $L_n$ of Definition~\ref{def:reg} has $2^n$ atoms and each atom meets the bounds of Proposition~\ref{prop:reg}.
\end{theorem}
\begin{proof}
The DFA for the atomic intersection $A_S$ is 
$\cD_S = (Q_S,\Sig,\Delta,(S,\ol{S}),F_S)$, where
 $F_S = \{ (X,Y) \mid  X \subseteq \{n\}, Y \subseteq Q_n\setminus \{n\}  \}$. 
 The transition semigroup of $\cD_n$ consists of all $n^n$ transformations of the state set $Q_n$.
Hence $(S,\ol{S})$ can be mapped to a final state in $F_S$ by taking a transformation that sends $S$ to $\{n\}$ and $\ol{S}$ to $\{1\}$. 
It follows that all $2^n$ atomic intersections $A_S$, $S \subseteq Q_n$ are atoms. By \distlemma, all distinct states in $\cD_S$ are distinguishable. It suffices to prove the number of reachable states in each $\cD_S$ meets the bounds.

 If $S=Q_n$,
 then $A_S$ is represented by $(Q_n,\emp)$, the reachable states of $\cD_S$ are of the form 
$(X,\emp)$, where $X$ is the image of $Q_n$ under some transformation in the transition semigroup.
Since we have all transformations, we can reach all $2^n-1$ states $(X,\emp)$, $\emp\subsetneq X\subseteq Q_n$. For $S = \emp$ a similar argument works.

If $\emp\subsetneq S\subsetneq Q_n$, then for any state $(X,Y)$ with $1\le X \le |S|$, $1\le Y \le n-|S|$ and $X\cap Y=\emp$, we can find a transformation mapping $S$ onto $X$ and $\ol{S}$ onto $Y$.
So all these states are reachable, and there are $\sum_{x=1}^{|S|}\sum_{y=1}^{n-|S|} \binom{n}{x}\binom{n-x}{y}$ of them.
In addition, $\bot$ is reachable from $(S,\ol{S})$ by the constant transformation $(Q_n \rightarrow 1)$ and so  the bound is met.
\qed
\end{proof}

\section{Complexity of Atoms in Right Ideals}

If $L$ is a right ideal, one of its quotients is $\Sig^*$; by convention we assume that $K_n=\Sig^*$. In any atom $A_S$ the quotient $K_n$ must be uncomplemented, that is, we must have $n\in S$. Thus $A_\emp$ is not an atom.
The results of this section were stated in~\cite{BrDa14} without proof; for completeness we include the proofs.
\begin{proposition}
\label{prop:bounds_right}
Suppose $L$ is a right ideal with $n\ge 1$ quotients. Then $L$ has at most $2^{n-1}$ atoms.
The  complexity $\kappa(A_S)$ of atom $A_S$ satisfies

\begin{equation}
	\kappa(A_S)\le
	\begin{cases}
	2^{n-1}, 	&\text{if $S=Q_n$;}\\
	1+ \sum_{x=1}^{|S|}\sum_{y=1}^{n-|S|} \binom{n-1}{x-1}\binom{n-x}{y},
			&\text{if $\emp \subsetneq S \subsetneq Q_n$.}
	\end{cases}
\end{equation}
\end{proposition}
\begin{proof}
Let $A_S$ be an atom.
Since $w^{-1}\Sig^*=\Sig^*$ for all $w\in\Sig^*$, $w^{-1}A_S$ always has $K_n$ uncomplemented; so if $(X,Y)$ corresponds to $w^{-1}A_S$, then $n\in X$.
Since the number of subsets $S$ of $Q_n$ containing $n$ is $2^{n-1}$, there are at most that many atoms. Consider two cases:
\be
\item
$S=Q_n$.
Then  $w^{-1}L=\bigcap_{i \in X} K_i$, and each such quotient of $A_S$ is represented by $(X,\emp)$, where $1\le |X|\le n$.
Since $n$ is always in $X$, there are at most $2^{n-1}$ quotients of this atom.
\item
$\emp \subsetneq S \subsetneq Q_n$. Then  $w^{-1} A_S = \bigcap_{i \in X} K_i \cap \bigcap_{i \in Y} \ol{K_i},$ where $1\le |X| \le |S|$ and $1\le |Y| \le n-|S|$.
We know that if $X\cap Y\neq \emp$, then $w^{-1}A_S=\emp$. 
Thus we are looking for pairs $(X,Y)$ such that $n\in X$ and $X\cap Y=\emp$.
To get $X$ we take $n$ and choose $|X|-1$ elements from $Q_n\setminus \{n\}$, and then to get $Y$ take $|Y|$ elements from $Q_n\setminus X$.
The number of such pairs is
$ \sum_{x=1}^{|S|}\sum_{y=1}^{n-|S|} \binom{n-1}{x-1}\binom{n-x}{y}$.
Adding the empty quotient we have our bound. \qed
\ee

\end{proof}

For $n=1$, $L=\Sig^*$ is a right ideal with one atom of complexity 1. 
For $n=2$, $L=aa^*$ is a right ideal with two atoms $L$ and $\ol{L}$ of complexity 2.
It was shown in~\cite{BrDa14} that the language of the DFA of Definition~\ref{def:rideal} is most complex in the sense that it meets all the bounds for common operations, but no proof of atom complexity was given. We include this proof here.

\begin{definition}
\label{def:rideal}
For $n\ge 3$, let $\cD_n=(Q_n,\Sig,\delta_n, 1, \{n\})$, where 
$\Sig=\{a,b,c,d\}$, 
and $\delta_n$ is defined by
$a= (1,\dots,n-1)$,
$b=(2,\ldots,n-1)$,
$c=(n-1\rightarrow 1)$
and $d=(n-1\rightarrow n)$.
Let $L_n$ be the language accepted by~$\cD_n$.
If $n=3$, $b$ is not needed; hence $\Sig=\{a,c,d\}$ suffices.
Also, let $L_2=aa^*$ and $L_1=a^*$. 
\end{definition}

\begin{figure}[ht]
\unitlength 7pt
\begin{center}\begin{picture}(42,10)(0,0)
\gasset{Nh=2.5,Nw=3.5,Nmr=1.25,ELdist=0.4,loopdiam=1.5}
\node(1)(1,7){1}\imark(1)
\node(2)(8,7){2}
\node(3)(15,7){3}
\node[Nframe=n](3dots)(22,7){$\dots$}
	{\small
\node(n-2)(29,7){$n-2$}
	}
	{\small
\node(n-1)(36,7){$n-1$}
	}
	{\small
\node(n)(43,7){$n$}\rmark(n)
	}
\drawloop(1){$b,c,d$}
\drawedge(1,2){$a$}
\drawloop(2){$c,d$}
\drawedge(2,3){$a,b$}
\drawloop(3){$c,d$}
\drawedge(3,3dots){$a,b$}
\drawedge(3dots,n-2){$a,b$}
\drawedge(n-2,n-1){$a,b$}
\drawedge(n-1,n){$d$}
\drawedge[curvedepth= 3,ELdist=-1.5](n-1,2){$b$}
\drawedge[curvedepth= 5,ELdist=-1.2](n-1,1){$a,c$}
\drawloop(n-2){$c,d$}
\drawloop(n){$a,b,c,d$}
\end{picture}\end{center}
\caption{DFA  of  a right ideal  whose atoms meet the bounds.}
\label{fig:RightIdeal}
\end{figure}
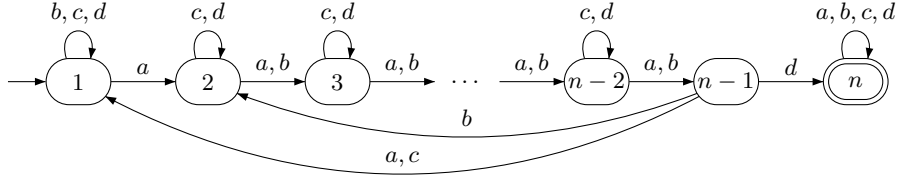

\begin{theorem}
For $n\ge 1$, the language $L_n$ of Definition~\ref{def:rideal} is a right ideal that has $2^{n-1}$ atoms and each atom meets the bounds of Proposition~\ref{prop:bounds_right}.
\end{theorem}
\begin{proof}
The cases $n<3$ are easily verified; hence assume $n\ge 3$.
By Proposition \ref{prop:piccard}, the transformations $\{a,b,c\}$ restricted to $Q_{n-1}$ generate all transformations of $Q_{n-1}$. When $d$ is included, we get all transformations of $Q_n$ that fix $n$. 
For $S \subseteq Q_n$, $n \in S$, consider the DFA $\cD_S$, which has initial state $(S,\ol{S})$. There is a transformation of $Q_n$ fixing $n$ that sends $(S,\ol{S})$ to the final state $(\{n\},\{1\})$. Hence $A_S$ is an atom if $n \in S$, and so $L_n$ has $2^{n-1}$ atoms.

We now count reachable and distinguishable states in the DFA of each atom. Suppose $S = Q_n$. The initial state of $\cD_S$ is $(Q_n,\emp)$; by transformations that fix $n$, we can reach any state $(X,\emp)$ with $\{n\} \subseteq X \subseteq Q_n$. There are $2^{n-1}$ such states, and since $A_X$ is an atom if $n \in X$, all of them are distinguishable by \distlemma.

Suppose $\emp \subsetneq S \subsetneq Q_n$. From the initial state $(S,\ol{S})$, by transformations that fix $n$ we can reach any $(X,Y)$ with $1 \le |X| \le |S|$, $1 \le |Y| \le n-|S|$, $n \in X$ and $X \cap Y = \emp$. 
There are $\sum_{x=1}^{|S|}\sum_{y=1}^{n-|S|} \binom{n-1}{x-1}\binom{n-x}{y}$ such states. For all such states $(X,Y)$, we have $n \in X$ and $n \in \ol{Y}$, so $A_X$ and $A_{\ol{Y}}$ are both atoms; hence by \distlemma, all of these states are distinguishable from each other and from $\bot$.
The state $\bot$ is also reachable by the constant transformation $(Q_n \to n)$, and so  the bound is met.
\qed
\end{proof}

\section{Complexity of Atoms in Left Ideals}
\label{sec:left}

If $L$ is a left ideal, then $L=\Sig^*L$, and $w^{-1}L$ contains $L$ for every $w\in \Sig^*$. By convention we let $L=K_1$.

\label{sec:left_bounds}
\begin{proposition}
\label{prop:bounds_left}
Suppose $L$ is a left ideal with $n\ge 2$ quotients. Then $L$ has at most $2^{n-1}+1$ atoms. The  complexity $\kappa(A_S)$ of atom $A_S$ satisfies
      
\begin{equation}
	\kappa(A_S) 
	\begin{cases}
		= n, 			& \text{if $S=Q_n$;}\\
		\le 2^{n-1},		& \text{if $S=\emp$;}\\
		 \le 1 + \sum_{x=1}^{|S|}\sum_{y=1}^{n-|S|}\binom{n-1}{x}\binom{n-x-1}{y-1},
		 			& \text{otherwise.}
		\end{cases}
\end{equation}
\end{proposition}
\begin{proof}
Consider the atomic intersections $A_S$ such that $1\in S$; then $\bigcap_{i\in S} K_i=L$ (since every quotient contains $L$), and there are two possibilities:
Either $S=Q_n$, in which case  $A_S=A_{Q_n}=\bigcap_{i\in Q_n} K_i=L$, or there is at least one quotient, say $K_i$ which is complemented. Since $K_i$ contains $L$, it can be expressed as $K_i=L\cup M_i$, where $L\cap M_i=\emp$. Then the intersection has the term $L\cap \ol{(L\cup M_i)}=\emp$, and $A_S$ is not an atom.
Thus for $A_S$ to be an atom, either  $1 \not \in S$ or $S=Q_n$. Hence there are at most $2^{n-1}+1$ atoms.

For atom complexity, consider the following cases:
\be
\item
$S=Q_n$.
Then  $A_{Q_n}=L$, and the  complexity of $A_{Q_n}$ is precisely $n$. 
\item
$S=\emp$. 
 Now  $A_\emp=\bigcap_{i \in Q_n} \ol{K_i}$, and
every quotient of $A_\emp$ is an intersection $\bigcap_{i \in Y} \ol{K_i}$, where 
$1\le |Y|\le |Q_n|$.
There are $2^n-1$ such intersections, but
consider any quotient $K_i\neq L$ of a left ideal; it can be expressed as $K_i=L\cup M_i$, where $L\cap M_i=\emp$. 
We have 
$$\ol{K_1}\cap \ol{K_i}=\ol{L} \cap \ol{L\cup M_i} = \ol{L}\cap \ol{L}\cap \ol{M_i}= \ol{L}\cap \ol{M_i}=\ol{K_i}.$$
Thus every  intersection $\bigcap_{i \in Y} \ol{K_i}$ which has $Y\neq\emp$ and does not have $\ol{K_1}$ as a term defines the same language as $\ol{K_1}\cap \bigcap_{i \in Y} \ol{K_i}$. 
There are $2^{n-1}-1$ such intersections. Adding 1 for the intersection which just has the single term $\ol{K_1}$, we get our bound $2^{n-1}$.

\item
$\emp \subsetneq S \subsetneq Q_n$. Then $A_S = \bigcap_{i \in S} K_i \cap \bigcap_{i \in \ol{S}} \ol{K_i}$, where neither $S$ nor $\ol{S}$ is empty.
If $1\in S$  this intersection  is empty, and so is not an atom.
Assume from now on that $1\not\in S$.
Every quotient of $A_S$ has the form $w^{-1} A_S = \bigcap_{i \in X} K_i \cap \bigcap_{i \in Y} \ol{K_i},$ where $1\le |X| \le |S|$ and $1\le |Y| \le n-|S|$.

	\be
	\item
	$1\in X$. We claim that $w^{-1} A_S=\emp$ for all $w \in \Sig^*$. For suppose that there is a term 		$K_i$, $i\in S$,  and a word $w\in \Sig^*$ such that $w^{-1}K_i=K_1$.
	Since $K_1\subseteq K_i$, we have $w^{-1}K_1 \subseteq w^{-1} K_i=K_1$.
	Since also $K_1\subseteq w^{-1}K_1$ because $L$ is a left ideal,  we have
	$w^{-1}K_1=K_1$. 
	But  $1\in \ol{S}$, so $w^{-1}\left(\bigcap_{i\in \ol{S}} \ol{K_i}\right) = \bigcap_{i \in Y} \ol{K_i}$ has $w^{-1}\ol{K_1} = \ol{K_1}$ as a term.
   Thus $1 \in Y$, which means $X \cap Y \neq \emp$.
	Hence  $w^{-1} A_S=\emp$.
	\item
	$1\not\in X$.
	We are looking for pairs $(X,Y)$ such that $X\cap Y=\emp$.
As we argued in (2), $\ol{K_1} \cap \ol{K_i} = \ol{K_i}$ for each $i$, so we can assume without loss of generality that $1 \in Y$.
	To get $X$ we choose $|X|$ elements from $Q_n\setminus\{1\}$ and to get $Y$ 	we take $\{1\}$ and choose $|Y|-1$ elements from $(Q_n\setminus X)\setminus 	\{1\}$.
	The number of such pairs is
	$ \sum_{x=1}^{|S|}\sum_{y=1}^{n-|S|}\binom{n-1}{x}\binom{n-x-1}{y-1}.$

		\ee
		Adding 1 for the empty quotient we have our bound. \qed
\ee
\end{proof}

Next we compare the bounds for left ideals with those for right ideals. To calculate the number of pairs $(X,Y)$ such that $n \in X$ and $X\cap Y=\emp$ for right ideals, we can first choose $Y$ from $Q_n\setminus \{n\}$ and then take $n$ and choose $|X|-1$ elements from $(Q_n\setminus Y)\setminus \{n\}$. The number of such pairs is
\[ 1 + \sum_{y=1}^{n-|S|}\sum_{x=1}^{|S|}\binom{n-1}{y}\binom{n-y-1}{x-1}.\]
If we interchange $x$ and $y$ we note that this is precisely the number of pairs $(X,Y)$ such that $1\in Y$ and $X\cap Y=\emp$ for an atom of a left ideal with a basis of size $n-|S|$.
Thus we have
\begin{remark}
\label{rem:symmetry}
Let $R$ be a right ideal of complexity $n$ and let $A_S$ be an atom of $R$, where $\emp \subsetneq S \subsetneq Q_n$.
Let $L$ be a left ideal of complexity $n$ and let $A'_{\ol{S}}$ be an atom of $L$. The upper bounds on the complexities of $A_S$ and $A'_{\ol{S}}$ are equal.
\end{remark}

Now we consider the question of tightness of the bounds in Proposition~\ref{prop:bounds_left}.
For $n=1$, $L=\Sig^*$ is a left ideal with one atom of complexity 1; so the bound of Proposition~\ref{prop:bounds_left} does not hold. 

The DFA of Definition~\ref{def:lideal} and Figure~\ref{fig:lideal} was introduced in~\cite{BrYe11}.
It was shown in~\cite{BrSz14} that the language of this DFA has the largest syntactic semigroup among left ideals of complexity $n$.
Moreover, it was shown in~\cite{BrLiu15} that this language also meets the bounds on the quotient complexity of boolean operations, concatenation and star. Together with our result about the number of atoms and their complexity, this shows that this language is the most complex left ideal.

\begin{definition}
\label{def:lideal}
For $n\ge 3$, let $\cD_n=(Q_n,\Sig,\delta_n, 1, \{n\})$, where 
$\Sig=\{a,b,c,d,e\}$, 
and $\delta_n$ is defined by
$a = (2,\dots,n)$,
$b =(2,3)$,
$c = (n\to 2)$,
$d = (n\to 1)$,
and $e=(Q_n\to 2)$.
If $n=3$, inputs $a$ and $b$ coincide; hence $\Sig=\{a,c,d,e\}$ suffices.
Also, let $\cD_2=(Q_2,\{a,b,c\},\delta_2, 1, \{2\})$, where 
$a = \one$,
$b = (Q_2\to 2)$,
$c = (Q_2\to 1)$. 
Let $L_n$ be the language accepted by~$\cD_n$; we have $L_2=\Sig^*b(a\cup b)^*$.
\end{definition}

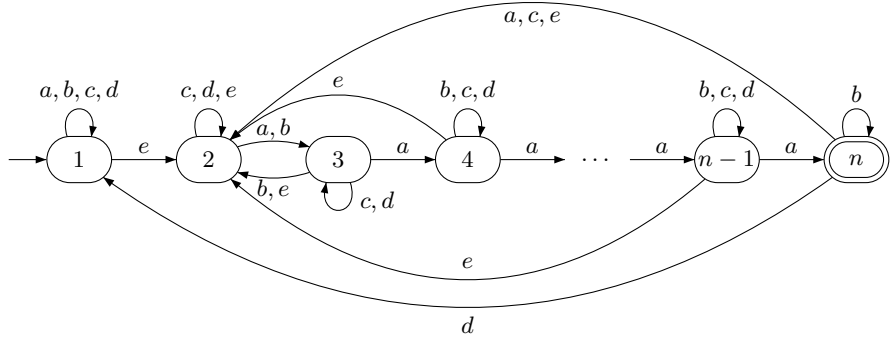
\begin{figure}[h]
\unitlength 7pt
\begin{center}\begin{picture}(43,17)(0,-1)
\gasset{Nh=2.5,Nw=3.5,Nmr=1.25,ELdist=0.4,loopdiam=1.5}
\node(1)(1,7){1}\imark(1)
\node(2)(8,7){2}
\node(3)(15,7){3}
\node(4)(22,7){4}
\node[Nframe=n](4dots)(29,7){$\dots$}
	{\small
\node(n-1)(36,7){$n-1$}
	}
	{\small
\node(n)(43,7){$n$}\rmark(n)
	}
\drawedge(1,2){$e$}
\drawloop(1){$a,b,c,d$}
\drawloop(2){$c,d,e$}
\drawedge[curvedepth= 1,ELdist=.1](2,3){$a,b$}
\drawedge[curvedepth= 1,ELdist=.1](3,2){$b,e$}
\drawloop[loopangle=270,ELpos=25](3){$c,d$}
\drawedge(3,4){$a$}
\drawedge[curvedepth= -3.5,ELdist=-1](4,2){$e$}
\drawedge(4,4dots){$a$}
\drawedge(4dots,n-1){$a$}
\drawloop(4){$b,c,d$}
\drawloop(n-1){$b,c,d$}
\drawedge(n-1,n){$a$}
\drawedge[curvedepth= 6.5,ELdist=-1.2](n-1,2){$e$}
\drawedge[curvedepth= -8.5,ELdist=.5](n,2){$a,c,e$}
\drawedge[curvedepth= 8.0,ELdist=.5](n,1){$d$}
\drawloop(n){$b$}
\end{picture}\end{center}
\caption{DFA  of a left ideal whose atoms meet the bounds.}
\label{fig:lideal}
\end{figure}

\begin{theorem}
For $n\ge 2$, the language $L_n$ of Definition~\ref{def:lideal} is a left ideal that has $2^{n-1}+1$ atoms and each atom meets the bounds of  Proposition~\ref{prop:bounds_left}.
\end{theorem}
\begin{proof}
It was proved in~\cite{BrYe11} that $L_n$ is a left ideal of complexity $n$.
The case $n=2$ is easily verified; hence assume $n\ge 3$.
It was proved in~\cite{BrSz14} that the transition semigroup of $\cD_n$ contains all transformations of $Q_n$ that fix 1 and all constant transformations.
Recall that if $A_S$ is an atom of a left ideal, then either $S = Q_n$ or $1 \not\in S$. 
For all $S$ with $1 \not\in S$, from $(S,\ol{S})$ we can reach the final state $(\{n\},\{1\})$ of $\cD_S$ (or $(\emp,\{1\})$ for $S = \emp$) by transformations that fix 1.
For $S = Q_n$, let $w = (Q_n \to n)$; then $(Q_n,\emp)w = (\{n\},\emp)$ is final in $\cD_S$.
Hence if $S = Q_n$ or $1 \not\in S$, then $A_S$ is an atom of $L_n$, and so $L$ has $2^{n-1}+1$ atoms.

We now count reachable and distinguishable states in the DFA of each atom. We know that $A_{Q_n}$ has complexity $n$ for all left ideals, so assume $1 \not \in S$.
If $S = \emp$, the initial state of $\cD_S$ is $(\emp,Q_n)$. By transformations that fix 1 we can reach $(\emp,Y)$ for all $Y$ with $\{1\} \subseteq Y \subseteq Q_n$. There are $2^{n-1}$ of these states. Since $\ol{Y}$ does not contain 1, $A_{\ol{Y}}$ is an atom, so all of these states are distinguishable by \distlemma.

If $\emp \subsetneq S \subsetneq Q_n$, the initial state of $\cD_S$ is $(S,\ol{S})$.
Since $1 \not\in S$, by transformations that fix 1, we can reach any state $(X,Y)$ with $1 \le |X| \le |S|$, $1 \le |Y| \le n-|S|$, $1 \not\in X$, $1 \in Y$, and $X \cap Y = \emp$.
There are $\sum_{x=1}^{|S|}\sum_{y=1}^{n-|S|}\binom{n-1}{x}\binom{n-x-1}{y-1}$ such states.
They are all distinguishable from each other and from $\bot$ by \distlemma, since $1 \not\in X$, $1 \in Y$ imply that $A_X$ and $A_{\ol{Y}}$ are both atoms.
We can also reach $\bot$ from $(S,\ol{S})$ by any constant transformation, and so  the bound is met. \qed
\end{proof}

\section{Complexity of Atoms in Two-Sided Ideals}
\subsection{Upper Bounds}
A language  is a two-sided ideal if it is both a right  ideal  and a left ideal.

\label{sec:2sided_bounds}
\begin{proposition}
\label{prop:bounds_2sided}
Suppose $L$ is a two-sided ideal with $n\ge 2$ quotients. Then $L$ has at most $2^{n-2}+1$ atoms. The  complexity $\kappa(A_S)$ of atom $A_S$ satisfies

\begin{equation}
	\kappa(A_S) 
	\begin{cases}
		= n, 			& \text{if $S=Q_n$;}\\
		\le 2^{n-2}+n-1,		& \text{if $S=Q_n\setminus\{1\}$;}\\
		 \le 1 + \sum_{x=1}^{|S|}\sum_{y=1}^{n-|S|}\binom{n-2}{x-1}\binom{n-x-1}{y-1},
		 			& \text{otherwise.}
		\end{cases}
\end{equation}

\end{proposition}
\begin{proof}
Since $L$ is a left ideal, $A_S$ is an atom only if $S = Q_n$ or $S \subseteq Q_n \setminus \{1\}$; since $L$ is a right ideal we must also have $n \in S$.
This gives our upper bound of $2^{n-2}+1$ atoms.

We know that $A_{Q_n}$ has complexity $n$ since this is true for left ideals.
Since $L$ is a right ideal, $A_\emp$ is not an atom, so we can assume $S \ne \emp$.

Suppose $A_S$ is an atom of $L$, with $S \ne Q_n$ and $S \ne Q_n \setminus \{1\}$.
We proved for left ideals that the number of distinct non-empty quotients of $A_S$ is bounded by the number of pairs $(X,Y)$, $1 \le |X| \le |S|$, $1 \le |Y| \le n-|S|$, $1 \not\in X$, $1 \in Y$, $X \cap Y = \emp$.
Since $L$ is a right ideal, we must also have $n \in X$ and $n \not\in Y$. 
There are $\binom{n-2}{|X|-1}$ possibilities for $X$, since $X$ must contain $n$ and the remaining $|X| - 1$ elements are taken from $Q_n \setminus \{1,n\}$.
If $X$ is fixed, there are $\binom{n-|X|-1}{|Y|-1}$ possibilities for $Y$, since $Y$ must contain 1 and the remaining $|Y|-1$ elements are taken from $(Q_n \setminus X) \setminus \{n\}$. Since $Q_n \setminus X$ always contains $n$,  the size of $(Q_n \setminus X) \setminus \{n\}$ is always $n -|X|-1$.
Summing over the possible sizes of $X$ and $Y$ and adding 1 for the empty quotient, we get the required bound.

This  leaves the case of $S = Q_n \setminus \{1\}$.
Each quotient of $A_S$ has the form
\[ w^{-1}A_S =  \left(\bigcap_{i\in X} K_i \right)\cap \ol{K_j},\]
where $K_j = w^{-1}K_1 = w^{-1}L$, and $n \in X$. We can view the non-empty quotients as states $(X,\{j\})$ of the DFA $\cD_S$ for $A_S$, where $\cD$ is a minimal DFA for $L$.
We must have $n \in X$ and $X \cap \{j\} = \emp$,  and so $j \not \in X$. Hence $\{n\} \subseteq X \subseteq Q_n \setminus \{j\}$, and there are $2^{n-2}$ choices for $X$.
However, for each $X$ there are potentially $n-1$ choices for $j$, giving an upper bound of $(n-1)2^{n-2}$ for the non-empty quotients, which is not tight. We need to look more carefully at the distinguishability relations between states of $\cD_S$.

For each $p$ in $Q_n$, define the set $S(p) = \{q \in Q_n \mid K_p \subsetneq K_q\}$. The elements of $S(p)$ are called the \emph{successors} of $p$.
Note that $p$ is not  a successor of itself.

Since $L$ is a left ideal, we have $L \subseteq K_i$ for all $i \in Q_n$.
It follows that $w^{-1}L = K_j \subseteq w^{-1}K_i$ for all $i \in Q_n$.
Thus in the formula for $w^{-1}A_S$ above, we have $K_j \subseteq K_i$ for all $i \in X$.
But if $K_j = K_i$ for any $i \in X$, then $w^{-1}A_S$ is empty. 
Thus $K_j \subsetneq K_i$ for all $i \in X$, which implies $X \subseteq S(j)$.

$X$ must contain $n$, since $L$ is a right ideal. Thus for each $j$, there are at most $2^{|S(j)|-1}$ distinguishable states $(X,\{j\})$.
The index $j$ can range from $1$ to $n-1$; if $j = n$ then $X \cap \{n\}$ is non-empty.
This gives an upper bound of $\sum_{j=1}^{n-1} 2^{|S(j)|-1}$ for the number of non-empty quotients. 

This bound still is not tight, so we refine it as follows.
Choose $i \ne n \in S(j)$ and a non-empty set $Y \subseteq S(i) \setminus \{n\}$. 
Then $K_i \subsetneq K_q$ for all $q \in Y$, so we have
$K_i \cap \left(\bigcap_{q \in Y} K_i \right) = K_i$.
This means $(\{i,n\},\{j\})$ is indistinguishable from $(Y \cup \{i,n\},\{j\})$.
Since $Y$ is non-empty and does not contain $n$, there are at most $2^{|S(i)|-1}-1$ possibilities for $Y$.

From this we get a new upper bound for the number of distinguishable states $(X,\{j\})$ for a fixed $j$, as follows: first take our previous bound of $2^{|S(j)|-1}$. Then for each $i \ne n \in S(j)$,   subtract $2^{|S(i)|-1}-1$ to account for the states $(Y \cup \{i,n\},\{j\})$ that are equivalent to $(\{i,n\},\{j\})$. Our new bound is
\[2^{|S(j)|-1} - \sum_{\substack{i \in S(j)\\ i \ne n}} (2^{|S(i)|-1} - 1) = 
(|S(j)| - 1) + 2^{|S(j)|-1} - \sum_{\substack{i \in S(j)\\ i \ne n}} 2^{|S(i)|-1}.\] 
Summing over all possible values of $j$, and adding 1 for the empty quotient, we get the following bound on the complexity of $A_S$:
\[1 + \sum_{j=1}^{n-1}\left((|S(j)| - 1) + 2^{|S(j)|-1} - \sum_{\substack{i \in S(j)\\ i \ne n}} 2^{|S(i)|-1}\right).\] 
Noting that $S(1) = \{2,\dotsc,n\}$ and $|S(1)| = n-1$, we pull out the $j = 1$ case from the outermost summation:
{\small
\[1 + (n-2) + 2^{n-2} - \sum_{\substack{i\in S(1)\\i \ne n}}2^{|S(i)|-1} + \sum_{j=2}^{n-1}\left((|S(j)| - 1) + 2^{|S(j)|-1} - \sum_{\substack{i\in S(j)\\i \ne n}} 2^{|S(i)|-1}\right).\]}
Observe that $1 + (n-2) + 2^{n-2}$ is equal to $2^{n-2} + n-1$, the bound we are trying to prove. We will show that the  value of the  rest of this formula is always less than or equal to zero. We pull 
$\sum_{j=2}^{n-1}2^{|S(j)|-1}$ out to the front:
{\small
\[ 2^{n-2} + n-1 + 
	\sum_{j=2}^{n-1}2^{|S(j)|-1} - 
	\sum_{\substack{i\in S(1)\\i \ne n}}2^{|S(i)|-1} + 
	\sum_{j=2}^{n-1}\left((|S(j)| - 1) - 
	\sum_{\substack{i \in S(j)\\ i \ne n}} 2^{|S(i)|-1}\right).
\]}
Note that  
$\sum_{j=2}^{n-1}2^{|S(j)|-1}= \sum_{\substack{i\in S(1)\\i \ne n}}2^{|S(i)|-1}$, so cancellation occurs:
\[ 2^{n-2} + n-1 + 
	\sum_{j=2}^{n-1}\left((|S(j)| - 1) - \sum_{\substack{i \in S(j)\\ i \ne n}} 2^{|S(i)|-1}\right).\]
Now, the value of the innermost summation is always greater than or equal to $|S(j)|-1$: for each $i \in S(j)$, $i \ne n$, we  know  that $n$ is a successor of $i$, and hence $S(i) \ge 1$ and $2^{|S(i)|-1} \ge 1$. Thus the value of the outermost summation is always less than or equal to zero.
It follows that the number of quotients of $A_S$ is at most $2^{n-2} + n-1$.
\qed
\end{proof}

Next we address the question of tightness of the bounds for two-sided ideals.\
For $n=1$, $L=\Sig^*$ is a two-sided ideal with one atom of complexity 1; so the bound of Proposition~\ref{prop:bounds_2sided} does not hold. 

The DFA of  Definition~\ref{def:2sided} and Figure~\ref{fig:2sided} was introduced in~\cite{BrYe11}.
It was shown in~\cite{BrSz14} that the language of the DFA of  Definition~\ref{def:2sided} has the largest syntactic semigroup among left ideals of complexity $n$.
 Moreover, it was shown in~\cite{BrLiu15} that this language also meets the bounds on the quotient complexity of boolean operations, concatenation and star. Together with our result about the number of atoms and their complexity, this shows that this language is the most complex two-sided ideal.

\begin{definition}
\label{def:2sided}
Let $n\ge 4$, and let  
$\cD_n =(Q_n,\Sig,\delta_n, 1,\{n\})$ be the DFA with
 $\Sig=\{a,b,c,d,e,f\}$,
$a= (2,3,\ldots,n-1)$,
$b= (2,3)$,
$c= (n-1\to 2)$,
$d= (n-1\to 1)$,
$e= (Q_{n-1}\to 2)$,
and $f= (2\to n)$.
For $n=4$,  inputs $a$ and $b$ coincide.
Also, let $\cD_3=(Q_3,\{a,b,c\},\delta_3, 1, \{3\})$, where
$a= \one$,
$b=(Q_2\to 2)$,
 $c=(2\to 3)$,
and 
let $\cD_2=(Q_2,\{a,b,c\},\delta_2, 1, \{2\})$, where 
$a= \one$,
$b=(Q_2\to 2)$,
$c=(Q_2\to 1)$. 
Let $L_n$ be the language accepted by $\cD_n$.
\end{definition}

\begin{figure}[h]
\unitlength 7pt
\begin{center}\begin{picture}(43,17)(0,-1)
\gasset{Nh=2.5,Nw=4,Nmr=1.25,ELdist=0.4,loopdiam=1.5}
\node(n)(8,14){$n$}\rmark(n)
\drawloop(n){$a,b,c,d,e,f$}
\drawedge(2,n){$f$}
\node(1)(1,7){1}\imark(1)
\node(2)(8,7){2}
\node(3)(15,7){3}
\node(4)(22,7){4}
\node[Nframe=n](4dots)(29,7){$\dots$}
	{\small
\node(n-2)(36,7){$n-2$}
	}
	{\small
\node(n-1)(43,7){$n-1$}
	}
\drawedge(1,2){$e$}
\drawloop(1){$a,b,c,d,f$}
\drawloop[loopangle=270,ELdist=.2](2){$c,d,e$}
\drawedge[curvedepth= 1,ELdist=.1](2,3){$a,b$}
\drawedge[curvedepth= 1,ELdist=-1.2](3,2){$b,e$}
\drawloop(3){$c,d,f$}
\drawedge(3,4){$a$}
\drawedge[curvedepth= 2.5,ELdist=-1](4,2){$e$}
\drawedge(4,4dots){$a$}
\drawedge(4dots,n-2){$a$}
\drawloop(4){$b,c,d,f$}
\drawloop(n-2){$b,c,d,f$}
\drawedge(n-2,n-1){$a$}
\drawedge[curvedepth= 5,ELdist=-1.2](n-2,2){$e$}
\drawedge[curvedepth= -9.5,ELdist=-1.2](n-1,2){$a,c,e$}
\drawedge[curvedepth= 9.5,ELdist=-1.5](n-1,1){$d$}
\drawloop(n-1){$b,f$}
\end{picture}\end{center}
\caption{ DFA of a  two-sided ideal whose atoms meet  the bounds.}
\label{fig:2sided}
\end{figure}
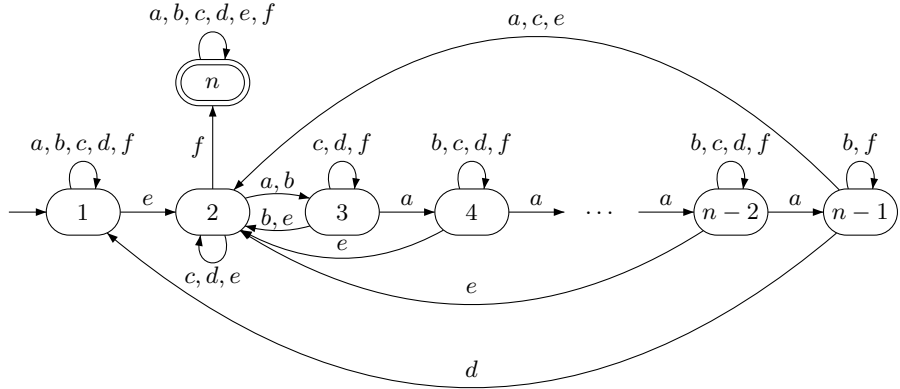

\begin{theorem}
For $n\ge 2$, the language $L_n$ of Definition~\ref{def:2sided} is a two-sided ideal that has $2^{n-2}+1$ atoms and each atom meets the bounds of Proposition~\ref{prop:bounds_2sided}.
\end{theorem}
\begin{proof}

It was proved in~\cite{BrYe11} that $L_n$ is a two-sided ideal of complexity $n$.
The cases with  $n<4$ are easily verified; hence assume $n\ge 4$.

The following observations were made in~\cite{BrSz14}:  
Transformations $\{a,b,c\}$ restricted to $Q_n \setminus \{1,n\}$ generate all the transformations of $\{2,\ldots,n-1\}$. Together with $d$ and $f$, they generate all transformations of $Q_n$ that fix $1$ and $n$.
Also, we have $ef = (Q_n \to n)$.

Recall that if $A_S$ is an atom of a two-sided ideal, then $n \in S$, and either $S = Q_n$ or $1 \not\in S$.
We know $A_{Q_n}$ is an atom of complexity $n$ for all left ideals (and hence all two-sided ideals), so assume $n \in S$, $1 \not\in S$.
Then $1 \in \ol{S}$,  and so from state $(S,\ol{S})$ in $\cD_S$ we can reach the final state $(\{n\},\{1\})$ by transformations that fix 1 and $n$.
Hence $A_S$ is an atom for every $S$ with $n \in S$, $1 \not\in S$. There are $2^{n-2}$ of these atoms, as well as the atom $A_{Q_n}$, for a total of $2^{n-2}+1$.

Consider the atom $A_S$ for $S \ne Q_n$ and $S \ne Q_n \setminus \{1\}$. 
In the DFA $\cD_S$, the initial state is $(S,\ol{S})$, and we have $n \in S$, $1 \not\in S$.
By transformations that fix $1$ and $n$, we can reach $(X,Y)$ for all $X,Y \subseteq Q_n$ such that $n \in X$, $1 \in Y$, $X \cap Y = \emp$, $1 \le |X| \le |S|$, $1 \le |Y| \le n-|S|$.
There are $\sum_{x=1}^{|S|}\sum_{y=1}^{n-|S|}\binom{n-2}{x-1}\binom{n-x-1}{y-1}$ such states.
Since $n \in X$, $1 \not \in X$ and $n \in \ol{Y}$, $1 \not\in \ol{Y}$ we see that $A_X$ and $A_{\ol{Y}}$ are atoms. Hence by \distlemma, all of these states are distinguishable from each other and from $\bot$.
Since $S \ne \emp$, we can reach $\bot$ from $(S,\ol{S})$ by $ef = (Q_n \to n)$. Hence  the  bound is met.

It remains to show that the complexity of $A_S$, $S = Q_n \setminus \{1\}$ also meets the bound.
The initial state of $\cD_S$ is $(\{2,\dotsc,n\},\{1\})$. By transformations that fix $1$ and $n$, we can reach all $2^{n-2}$ states of the form $(X,\{1\})$ with $\{n\} \subseteq X \subseteq Q_n \setminus \{1\}$. From $(\{n\},\{1\})$, we can reach $n-2$ additional states $(\{n\},\{i\})$ for $2 \le i \le n-1$ by $ea^{i-2}$.
Finally, we can reach the sink state $\bot$ from the initial state by $ef = (Q_n \to n)$. This gives a total of $2^{n-2} + n-1$ reachable states, which matches the upper bound.

To see these states are distinguishable, note that $A_X$ is an atom if $\{n\} \subseteq X \subseteq Q_n \setminus \{1\}$. Also, $A_{\ol{\{1\}}} = A_{Q_n \setminus \{1\}}$ is an atom. Hence by \distlemma, all states of the form $(X,\{1\})$ are distinguishable from each other and from $\bot$.
Also, $(\{n\},\{i\})$ is distinguished from $(\{n\},\{j\})$ by $a^{n-i}f$, which sends the former state to the non-final state $\bot$, but sends the latter to some final state $(\{n\},\{k\})$ with $k \ne 2$. And each $(\{n\},\{j\})$, $1 \le j \le n-1$ is a final state, so it is distinguishable from all states of the form $(X,\{1\})$, $X \ne \{n\}$ and from $\bot$, since they are not final.
Hence all $2^{n-2}+n-1$ reachable states are distinguishable.
\qed
\end{proof}

\section{Some Numerical Results}
The following tables compare the maximal complexities for atoms $A_S$ of two-sided ideals (first entry), left ideals (second entry) and regular languages (third entry) with complexity $n$. Right ideals are omitted because their complexities are essentially the same as those of left ideals, by Remark \ref{rem:symmetry}.
When the maximal complexity is undefined (e.g., because no languages in a class have atoms $A_S$ for a particular size of $S$) this is indicated by an asterisk. The maximum values for each $n$ are in boldface. The $n^{\rm{th}}$ entry in the \emph{ratio} row shows the approximate value of $m_n/m_{n-1}$, where $m_i$ is the $i^{\rm{th}}$ entry in the \emph{max} row.

\[
\begin{array}{|c|c|c|c|c|c|c|}
\hline
\ n\ & 1\ &\ 2\ &\ 3\ &\ 4\ &\ 5\ &\ \cdots\\
\hline
\hline
|S|=0 & \ast/{\bf 1}/{\bf 1} &\ast/{\bf 2}/{\bf 3} &\ast/4/7 & \ast/8/15 & \ast/16/31 & \cdots\\
\hline
|S|=1 & {\bf 1}/{\bf 1}/{\bf 1} &{\bf 2}/{\bf 2}/{\bf 3} &3/{\bf 5}/{\bf 10} &5/13/29 & 9/33/76 & \cdots\\
\hline
|S|=2 &  & {\bf 2}/{\bf 2}/{\bf 3} &{\bf 4}/4/{\bf 10} &{\bf 8}/{\bf 16}/{\bf 43} &{\bf 20}/{\bf 53}/{\bf 141} &\cdots\\
\hline
|S|=3 &  &  & 3/3/7 & 7/8/29 & {\bf 20}/43/{\bf 141} &\cdots\\
\hline
|S|=4 &  &  &  & 4/4/15 & 12/16/76 & \cdots\\
\hline
|S|=5 &  &  &  &  & 5/5/31 & \cdots\\
\hline
\textit{max} & 1/1/1 & 2/2/3 & 4/5/10 & 8/16/43 & 20/53/141 & \cdots\\
\hline
\textit{ratio} & - & 2.00/2.00/3.00 & 2.00/2.50/3.33 & 2.00/3.20/4.30 & 2.50/3.31/3.28 & \cdots\\
\hline
\end{array}
\]

{\small
\[
\begin{array}{|c|c|c|c|c|}
\hline
\ n\ & 6\ &\ 7\ &\ 8\ &\ 9\ \\
\hline
\hline
|S|=0 & \ast/32/63 & \ast/64/127 & \ast/128/255 & \ast/256/511 \\
\hline
|S|=1 & 17/81/187 & 33/193/442 & 65/449/1,017 & 129/1,025/2,296 \\
\hline
|S|=2 & 48/156/406 & 112/427/1,086 & 256/1,114/2,773 & 576/2,809/6,859 \\
\hline
|S|=3 & {\bf 64}/{\bf 166}/{\bf 501} &{\bf 182}/{\bf 542}/{\bf 1,548} &484/1,611/4,425 & 1,234/4,517/12,043 \\
\hline
|S|=4 & 48/106/406 & {\bf 182}/462/{\bf 1,548} &{\bf 584}/{\bf 1,646}/{\bf 5,083} &{\bf 1,710}/{\bf 5,245}/{\bf 15,361} \\
\hline
|S|=5 & 21/32/187 & 112/249/1,086 & 484/1,205/4,425 & {\bf 1,710}/4,643/{\bf 15,361} \\ 
\hline
|S|=6 & 6/6/63 & 38/64/442 & 256/568/2,773 & 1,234/3,019/12,043 \\
\hline
|S|=7 &  & 7/7/127 & 71/128/1,017 & 576/1,271/6,859 \\
\hline
|S|=8 &  &  & 8/8/255 & 136/256/2,296 \\
\hline
|S|=9 &  &  &  & 9/9/511 \\
\hline
\textit{max} & 64/166/501 & 182/542/1,548 & 584/1,646/5,083 & 1,710/5,245/15,361 \\
\hline
\textit{ratio} &  3.20/3.13/3.55 & 2.84/3.27/3.09 & 3.21/3.04/3.28 & 2.93/3.19/3.02 \\
\hline
\end{array}
\]
}

\section{Conclusions}
We have derived tight upper bounds for the number of atoms and quotient complexity of atoms in right, left and two-sided regular ideal languages.
The recently discovered relationship between atoms and the Myhill and Nerode congruence classes opens up many interesting research questions. 
The quotient complexity of a language is equal to the number of Nerode classes, and 
the number of Myhill classes has also been used as a measure of complexity, called \emph{syntactic complexity} since it is equal to the size of the syntactic semigroup. 
We can view the number of atoms as a third fundamental measure of complexity for regular languages.

It is known~\cite{BrTa13} that the number of atoms of a regular language $L$ is equal to the quotient complexity of the \emph{reversal} of $L$. The quotient complexity of reversal has been studied for various classes of languages in the context of determining the quotient complexity of operations on regular languages. Hence, the maximal number of atoms is known for many language classes.

However, as far as we know the \emph{quotient complexity} of atoms has not been studied outside of regular languages and ideals.
For simplicity, let us call the atom congruence the \emph{left congruence}, the Nerode congruence the \emph{right congruence}, and the Myhill congruence the \emph{central congruence}. 
When computing the quotient complexity of atoms, we are computing the number of \emph{right congruence classes} of each \emph{left congruence class}. 
We can consider other permutations of this idea: how many right classes and left classes do the central classes have? How many central classes do the left classes have? 
These questions are outside the scope of this paper, but we believe they should be investigated.

\providecommand{\noopsort}[1]{}

\end{document}